\documentclass[conference,a4,10pt]{IEEEtran}
\usepackage{pgfplots}
\usepackage{cite}
\ifCLASSINFOpdf
\fi
\usepackage{amsmath,amsfonts,amsthm} 
\usepackage{amssymb}
\usepackage{bbm}
\usepackage{float}
\usepackage{graphicx}
\usepackage{algorithm}
\usepackage{algpseudocode}
\usepackage{subcaption}
\usepackage{siunitx}
\usepackage{tikz}
\usetikzlibrary{arrows,automata}
\usepackage{multicol}
\usepackage{comment}

\usepackage{pgfplots}
\usepackage{epstopdf}
\usepackage{multirow}
\usepackage[font=scriptsize,skip=5pt]{caption}
\usepackage{soul}
\usepackage{todonotes}
\usepackage{color, colortbl}
\usepackage{array}
\usetikzlibrary{patterns}
\usepackage{balance}

\usepackage{enumitem}
\usepackage{footmisc}
\allowdisplaybreaks
\usetikzlibrary{shapes.geometric}

\graphicspath{{images/}}
\usepackage{balance}
\newtheorem{theorem}{Theorem}

\makeatletter
	{\end{proof}%
	\renewcommand{\qedsymbol}{\qedsymbol}}
\makeatother

\graphicspath{{images/}}

\DeclareMathAlphabet\mathbfcal{OMS}{cmsy}{b}{n}

\definecolor{Gray}{gray}{0.9}
\definecolor{LightCyan}{rgb}{0.88,1,1}

\algnewcommand{\IIf}[1]{\State\algorithmicif\ #1\ \algorithmicthen}
\algnewcommand{\IElse}[2]{\State\algorithmicelse\ #2\ }
\algnewcommand{\EndIIf}{\unskip\ \algorithmicend\ \algorithmicif}

\IEEEoverridecommandlockouts
\begin{document}
	\setlength{\parskip}{0em}
	
	\title{Importance-Aware Fresh Delivery of Versions over Energy Harvesting MACs
 \thanks{This work was supported by the Science and Engineering Research Board, Department of Science and Technology, Government of India, under Project SRG/2020/001545.}}
		 
\author{\IEEEauthorblockN{Gangadhar Karevvanavar\IEEEauthorrefmark{1},
		Rajshekhar V Bhat\IEEEauthorrefmark{2}}
	
	\IEEEauthorblockA{
		\IEEEauthorrefmark{1}\IEEEauthorrefmark{2}Indian Institute of Technology Dharwad, Dharwad, Karnataka,  India\\
	 \IEEEauthorrefmark{1}212021007@iitdh.ac.in, \IEEEauthorrefmark{2}rajshekhar.bhat@iitdh.ac.in} 
}

	\maketitle
	\IEEEpeerreviewmaketitle
\begin{abstract}
We consider a scenario  where multiple users, powered by energy harvesting, send version updates over a fading multiple access channel (MAC) to an access point (AP). Version updates having random importance weights arrive at a user according to an exogenous arrival process, and  a new version renders all previous versions obsolete.
As   energy harvesting imposes a time-varying  peak power constraint, it is  not possible to deliver all the bits of a version instantaneously.
Accordingly, the AP chooses the objective of minimizing a finite-horizon time average  expectation of the product of importance weight  and a convex increasing function of the number of remaining bits of a version to be transmitted at each time instant. The objective enables \emph{importance-aware} delivery of  as many bits, as soon as possible.  
In this setup, the AP optimizes the objective function subject to an achievable rate-region constraint of the MAC and energy constraints at the users, by  deciding the transmit power and the number of bits to be transmitted by each user. We obtain a Markov Decision Process (MDP)-based optimal online policy to the problem and derive structural properties of the policy. 
We then develop a neural network (NN)-based online heuristic policy, for which we train an NN on the optimal offline policy derived for different sample paths of energy, version arrival and channel power gain processes. 
Via numerical simulations, we observe that the NN-based online  policy performs competitively with respect to the  MDP-based online policy. 

\end{abstract}

\section{Introduction}
Many upcoming applications of beyond fifth generation  ultra-reliable and low-latency communications  and sixth generation  communication networks require energy-efficient importance/context/semantic/goal-aware fresh delivery of information from sources to destinations \cite{Sem-NP,Sem-WC,Sem-Popovski}. Towards achieving this, efforts in several directions are being carried out. This includes the  proposal and optimization of metrics  for measuring freshness of information,  such as  age of information (AoI) and its variants including version AoI \cite{AoG, VAoI}, urgency of information \cite{UoI} and age of incorrect information/estimates \cite{AoII,AoIE}. The efforts also include 
development of theory for semantic communications \cite{Sem-Theory}, deriving semantic-aware transmission policies \cite{Sem-Game,Sem-Geofrrey,Sem-bitrate} and semantic-aware communication of text \cite{Sem-NLP}, speech \cite{Sem-Speech} and image \cite{Sem-Image} data. Moreover, due to the expected large number of deployments of communication devices, most emerging networks may be powered from energy harvesting sources. Accordingly, AoI and  semantic-aware communications in energy harvesting systems  has also been considered \cite{Sem-EH-Poor-ICASSP,Sem-EH-Sun-ICASSP,AoI-ICC-EH,AoI-EH-Ulukus,AoI-IT-Harpreet, Sem-EH-ISIT}, including   communications over point-to-point \cite{BhatP2P}, multiple access \cite{AoI-MAC-NOMABhat, AoI-MAC-NOMA} and broadcast \cite{AoI-ICC-EH} channels.  

Motivated by the above works, in this work, we consider a multiple access channel (MAC) where the versions of fixed size but random importances arrive at random time instants at the users, equipped with batteries that are replenished using energy harvesting sources. The previous version becomes obsolete when a new version arrives. In many practical applications, such as real-time video communication, even a small  number of bits of a version can provide useful description of the version; the more the number of bits transmitted (or equivalently, the smaller the number of un-transmitted bits), the better is the description \cite{layered}.  
Moreover, the more delayed the transmission of a set of bits is, the less useful is the description derived from them. Hence, for achieving more useful description of a version with less delay at the AP,  the users are required to deliver as many bits, as soon as possible.   
To achieve this, we consider and optimize the following metric:  \emph{the product of the importance weight and a convex increasing function of the number of remaining bits of a version to be transmitted}. Due to convexity and increasing behavior of the function, optimizing it across slots enables delivery of more bits with less delay. 

The existing literature considers importance-aware  communications, where the importance of versions are derived implicitly  \cite{Sem-AoII-Ephremides, AoIE}  or considered to be specified by the external source that generates the versions   \cite{UoI}.  In this work, we consider the latter, i.e., we consider that the exogenous source itself specifies the importance weight for each version it generates.   We also consider that the users communicate over a fading  MAC using the non-orthogonal multiple access (NOMA), as in \cite{AoI-MAC-NOMABhat},   where more than one user can transmit information simultaneously, subject to an achievable rate region constraint of the MAC.

Optimization of communication systems for fresh delivery of information has been approached via different frameworks,  including the Markov Decision Process (MDP) \cite{AoIE}, reinforcement learning (RL) \cite{Sem-RF} and Lyapunov  optimization \cite{AoPI}. In this work, we adopt the MDP framework for obtaining the optimal online policy. We also propose a neural network (NN)-based policy, motivated by \cite{ICASSP_Mohit}. The work \cite{ICASSP_Mohit} first obtains an offline optimal policy for a throughput maximization problem in an energy harvesting MAC,  assuming complete non-causal knowledge of realizations of random processes involved. Then, it trains a neural network (NN) to output the instantaneous transmit power values for given state of the system, which includes the state of the batteries and channel power gains. The authors show that the NN-based policy performs competitively to the online optimal policy obtained using the MDP framework. Along the similar lines, in this work, we consider an offline policy and train an NN to obtain a heuristic online policy.  
As mentioned, in this work, we consider  the problem of  minimizing the finite-horizon expected average of the product of importance weight  and a convex increasing function of the number of un-transmitted bits of  versions in a MAC using NOMA, subject to an achievable rate region and variable peak energy constraints occurring due to the energy harvesting power supply. The  decision variables are the transmit power and the number of bits to be transmitted at each user and time instant.  Towards solving this, our  main contributions are the following: 
\begin{itemize}
	\item We obtain an MDP-based optimal online policy, via iterative minimization of the  Bellman equations subject to instantaneous achievable rate region and energy constraints, with states being the amount of energy available, the number of bits to be transmitted, channel power gains and the importance weights of all the users, and the actions being the transmit power and the number of bits to be transmitted. 	We show that at each iteration, the minimization of Bellman equation is a convex optimization problem. Using this result, we showed that the optimal transmit power and the number of bits to be transmitted by  a user are non-decreasing in the amount of energy available in the battery  (respectively, the number of remaining bits) for a fixed number of remaining bits to be transmitted (respectively, fixed amount of energy available in the battery), when all other terms in the states and actions are fixed. 
	\item We then develop an NN-based online heuristic policy, for obtaining which  we train a neural network on the optimal offline policy derived for different sample paths of energy,  version arrival and channel power gain processes. The neural network takes the instantaneous states as the input and outputs the actions of all the users.   The considered problem in the offline case happens to be a convex optimization problem that can be solved easily using standard numerical techniques. This helps in  generating a large amount of data required for training the neural network.  
 
	\item Via numerical simulations, we show that the NN-based policy performs competitively with respect to the MDP-based  optimal online policy obtained by exploiting the derived structural properties. 
	
\end{itemize}
 
The remainder of the paper is organized as follows. We
present the system model and problem formulation in  Section~\ref{sec:sys_model} and propose the optimal online policy and  an NN-based heuristic online policy derived via the optimal offline policy in  Section~\ref{sec:solution}. We present
numerical results in Section~\ref{sec:num} and conclude in Section \ref{sec:conclusions}.

\section{System Model and Problem Formulation}\label{sec:sys_model}
We consider $M$ users communicating version updates to a common AP over a block-fading MAC using NOMA strategy. The system is assumed to be time-slotted, where each slot (block) is of unit time duration and the slots are indexed by $t\in \{1,2,\ldots,T\}$ for some finite positive integer, $T$. The users are powered from energy harvesting sources. In this section, we present energy arrival and channel models,   the rate-power relationship, the version arrival model, a description of the system operation and the problem formulation.

\subsection{Energy and Channel Model and Rate-Power Relationship}
At the start of slot $t$, a packet of harvested energy of $E_i(t)$ units arrives at user $i\in \mathcal{M}\triangleq \{1,2,\ldots,M\} $, where $E_i(1),E_i(2),\ldots,$ is a sequence of independent and identically distributed (i.i.d.) random variables having distribution identical to a random variable $E_i$, distributed  over a finite set, $\mathcal{E}_i\subset\mathbb{R}^+$, where $\mathbb{R}^+$ is a set of non-negative real numbers. We model the wireless channel as a block-fading channel, where the channel power gain of a user remains constant over a block  and changes across the blocks independently, as in \cite{Bhat17,Bhat19}. 
In slot $t$, let the channel power gain  be $H_i(t)$, where  ${H_i(1), H_i(2),\ldots,}$ is a  sequence of i.i.d. random variables distributed as $H_i$ over a finite set,  $\mathcal{H}_i\subset\mathbb{R}^+$.  
Suppose the transmit power of user $i\in \mathcal{M}$ is $P_i$ and the channel power gain is $h_i$, then any rate tuple, $(\rho_1,\ldots,\rho_M)$ that satisfies
\begin{align}
\sum_{i\in \mathcal{S}}\rho_i \leq g\left(\sum_{i\in \mathcal{S}}h_iP_i\right),
\end{align}
for all $\mathcal{S}\subseteq \mathcal{M}$ is achievable by adopting superposition coding at the users and successive interference cancellation at the AP.  
Here, $g(\cdot)$ is a concave, monotonically increasing function and $g(0)=0$.

\subsection{Version Arrival Model}
Let $A_i(t)$ be an indicator variable which takes value $1$ if the $i^{\rm th}$ source device gets a new version and zero otherwise, in slot $t$, where  $A_i(1), A_i(2),\ldots,$ forms a sequence of i.i.d. Bernoulli random variables. Let $w_i(t)$ be the importance weight associated with the version that arrives in slot $t$. We consider that  $w_i(t) = w_i(t-1)(1-A_i(t)) + WA_i(t)$ for all $t=1,2,\ldots,T$,  where $W$ is a discrete random variable distributed over a finite set, $\mathcal{W}\subset \mathbb{R}^+$, and $w_{i}(0)=0$. 

Each version consists of  $r_{\rm max}$ bits. The arrived bits may be transmitted over multiple slots until a next version arrives, at which time, the current version becomes obsolete and  all its remaining bits are discarded and replaced with the fresh $r_{\rm max}$ bits. 
Let $\rho_i(t)$ be the number of bits transmitted in slot $t$. Then, the number of remaining bits to be transmitted at the start of slot $t$, evolves as follows: 
\begin{align}\label{eq:bit-evolution}
r_i(t)= 
\begin{cases}
r_i(t-1)- \rho_i(t-1)& \text{if $A_i(t)=0$},\\
r_{\rm max}            & \text{if $A_i(t)=1$}, 
\end{cases}
\end{align}
where $r_i(0)=0$ and $\rho_i(t) \leq r_i(t)$ for all $t=1,2,\ldots,T$ and $i\in \mathcal{M}$. 
Our objective is minimizing the sum of long-term average expectation of the product of the importance weight and a convex increasing function of the number of remaining bits to be transmitted, across all the users. That is,  our objective function is given by $(1/TM)\sum_{t=1}^T\sum_{i=1}^{M} \mathbb{E}[w_i(t)f(r_i(t)-\rho_i(t))]$, where $f(\cdot)$ is a convex increasing function of its argument, with $f(0)\geq 0$ and $f(r_{\rm max})=1$.   

\subsection{System Operation and Constraints}
At the start of  slot $t$, the  AP observes the channel power gain realization, $\mathbf{h}(t) =(h_1(t),\ldots,h_M(t))$, the importance weights, $\mathbf{w}(t) =(w_1(t),\ldots,w_M(t))$ the amount of energy stored in the batteries, $\mathbf{B}(t) = (B_1(t),\ldots, B_M(t))$ and the number of remaining bits, $\mathbf{r}(t) = (r_1(t),\ldots, r_M(t))$. It then decides the  transmit power  $\mathbf{P}(t) = (P_1(t),\ldots, P_M(t))$ and the number of bits to be transmitted $\boldsymbol{\rho}(t) = (\rho_1(t),\ldots, \rho_M(t))$. The battery evolves as follows:  
\begin{align}\label{eq:battery-evolution}
&B_{i}(t) =\min(B_{i}(t-1)-P_i(t-1)+E_i(t), B_{\rm max}),
\end{align} 
for $i\in \mathcal{M}$ and $t \in \{1,2,\ldots,T\}$, where   $B_i(0)=0$ and $P_i(t)\leq B_i(t)$. Here,  $B_i(t)$ is the amount of energy stored in the battery at the start of slot $t$, and $B_{\rm max}$ is the maximum amount of energy that can be stored in the battery of the users.

\subsection{Problem Formulation}
In this work, we are interested in minimizing 
the finite-horizon expected average of the product of importance weight  and a convex increasing function of the number of un-transmitted bits of  versions which can be accomplished by solving the following optimization problem: 
\begin{subequations}\label{eq:main-opt-problem}
\begin{align}
\underset{\pi}{\text{min}} &\;\;\;\; \frac{1}{TM}\sum_{t=1}^{T}\sum_{i=1}^{M}\mathbb{E}[w_i(t)f(r_i(t) - \rho_i(t))],   \;\;\label{eq:AoI}\\
	\text{subject to}&\; \;\eqref{eq:bit-evolution}, \eqref{eq:battery-evolution},\nonumber \\
	&\;\; \sum_{i\in \mathcal{S}}\rho_i(t) \leq g\left(\sum_{i\in \mathcal{S}}h_i(t)P_i(t)\right),\; \forall\mathcal{S}\subseteq \mathcal{M},\label{eq:capacity}\\
	&\;\; P_i(t) \leq B_i(t),\; \forall i\in \mathcal{M},\label{eq:power}\\
	&\;\; \rho_i(t)\leq r_i(t),\;  \forall i\in \mathcal{M},\label{eq:bits}
	\end{align}
\end{subequations}
for all $t\in\{1,2,\ldots,T\}$, where \eqref{eq:capacity} is the instantaneous MAC achievable rate region constraint and the constraints \eqref{eq:power} and \eqref{eq:bits} self-explanatory. The policy $\pi$ is the decision rule to select $P_i(t)$ and $\rho_i(t)$ for each user $i\in \mathcal{M}$ and time instant $t\in \{1,2,\ldots,T\}$. The expectation in \eqref{eq:AoI} is with respect to the randomness in energy arrival, channel power gain and version arrival processes. 
In the current work, we only consider the finite-horizon case, with a finite $T$. 

\section{Solution}\label{sec:solution}
In this section, we provide various solutions to  \eqref{eq:main-opt-problem}. We first cast \eqref{eq:main-opt-problem} under the MDP framework and obtain structural properties of the optimal solution. We then obtain an NN-based heuristic online policy, that is trained using the optimal solution in the offline case when all the future realizations of the random processes are known non-causally. 
\subsection{Optimal Online Solution and Structural Properties}
The optimization problem in \eqref{eq:main-opt-problem} has form of an MDP.  We consider each tuple $s(t)$, where $s(t) = (\mathbf{B}(t), \mathbf{r}(t), \mathbf{h}(t),  \mathbf{w}(t))\in [0,B_{\rm max}]^M \times [0, r_{\rm max}]^M \times \mathcal{H}^M \times \mathcal{W}^M$  as a state, and an action is given by $a(t) = (\mathbf{P}(t), \boldsymbol{\rho}(t))$ at time $t\in \{1,2,\ldots,T\}$. The actions are subject to the constraints in \eqref{eq:capacity}, \eqref{eq:power} and \eqref{eq:bits} at each time instant. The instantaneous cost associated with the MDP is  $\sum_{i=1}^{M}\mathbb{E}[w_i(t)f(r_i(t)-\rho_i(t))]$ at time $t$.  The state transition probability matrix can be generated based on the distributions of the energy and version arrivals and the channel model.

Given an initial state, $s_1=(\mathbf{B}_1, \mathbf{r}_1, \mathbf{h}_1, \mathbf{w}_1)$, the optimal objective value, given by ${1}/(MT)V_1(s_1)$, can be computed via backward recursion \cite{Rui-MDP}, starting from time step $t=T$, $t=T-1$ and so on. For $t = T$, we solve the following Bellman equation subject to constraints:
\begin{subequations}\label{eq:atT}
	\begin{align}
	&V_T(\mathbf{B}(T), \mathbf{r}(T),\mathbf{h}(T), \mathbf{w}(T))=\nonumber&&\\
	&\min_{\mathbf{P}(T), \boldsymbol{\rho}(T)}\sum_{i=1}^{M}\mathbb{E}[w_i(T)f(r_i(T) - \rho_i(T))],\label{eq:atTO}&&\\
	&\text{subject to} \; \sum_{i\in \mathcal{S}}\rho_i(T) \leq g\left(\sum_{i\in \mathcal{S}}h_i(T)P_i(T)\right), &&\\
	&\;\;\qquad \qquad 0\leq P_i(T) \leq B_i(T),\; \forall i\in \mathcal{M},\label{eq:batT}&&\\
	&\;\; \qquad \qquad 0\leq \rho_i(T)\leq r_i(T),\; \forall i\in \mathcal{M},&&
	\end{align}
\end{subequations}
and for $t=T-1,T-2,\ldots,1$, we solve the following: 
\begin{subequations}\label{eq:at_t}
\begin{align}
&V_{t}(\mathbf{B}(t), \mathbf{r}(t),\mathbf{h}(t), \mathbf{w}(t))=\nonumber&&\\
&\min_{\substack{\mathbf{P}(t),\boldsymbol{\rho}(t)}}Q_{t}\left(\mathbf{B}(t), \mathbf{r}(t),\mathbf{h}(t), \mathbf{w}(t),\mathbf{P}(t),\boldsymbol{\rho}(t)\right),&&\\
&\text{subject to}\; \sum_{i\in \mathcal{S}}\rho_i(t) \leq g\left(\sum_{i\in \mathcal{S}}h_i(t)P_i(t)\right),&&\label{eq:cap1}\\
&\;\;\qquad\qquad 0\leq P_i(t) \leq B_i(t),\; \forall i\in \mathcal{M},&&\\
&\;\; \qquad \qquad 0\leq \rho_i(t)\leq r_i(t),\; \forall i\in \mathcal{M},  && 
\end{align}
\end{subequations}
for all $\mathcal{S}\subseteq \mathcal{M}$, where 
\begin{align}\label{eq:Qfunction}
&Q_{t}(\mathbf{B}(t), \mathbf{r}(t),\mathbf{h}(t), \mathbf{w}(t),\mathbf{P}(t),\boldsymbol{\rho}(t))\nonumber\\
& =  \sum_{i=1}^{M}\mathbb{E}[w_i(t)f(r_i(t)-\rho_i(t))]\nonumber\\
&\quad+\mathbb{E}[V_{t+1}\left(\mathbf{B}(t+1),\mathbf{r}(t+1), \mathbf{h}(t+1), \mathbf{r}(t+1)\right)], 
\end{align}
for $t\in\{1,2,\ldots,T-1\}$, where  $r_i(t)$ and $B_i(t)$ evolve according to   \eqref{eq:bit-evolution} and \eqref{eq:battery-evolution}, respectively.  
%

We have the following result. 
\begin{theorem}\label{thm:convex}
The problem \eqref{eq:atT} for $t=T$ and \eqref{eq:at_t} for each $t\in \{1,\ldots,T-1\}$ are  convex. Moreover, for each $t\in \{1,2,\ldots,T\}$, $V_{t}(\mathbf{B}(t), \mathbf{r}(t),\mathbf{h}(t), \mathbf{w}(t))$ is convex in $(\mathbf{B}(t), \mathbf{r}(t))$. 
\end{theorem}
\begin{proof}
	We prove the theorem by induction on the backward recursion  mentioned above. At time step $T$, we solve  \eqref{eq:atT}. In \eqref{eq:atT}, the objective function   $f(\cdot)$ is convex in its argument and the argument is an  affine decreasing function of $\boldsymbol{\rho}(T)$. Since affine transformation of variables preserves convexity \cite{boyd2004convex},  \eqref{eq:atTO}  is convex. Moreover, all the constraints are trivially convex. Hence, \eqref{eq:atT} is a convex optimization problem with optimal objective value $V_T(\mathbf{B}(T), \mathbf{r}(T),\mathbf{h}(T), \mathbf{w}(T))$. Now, since the optimal objective value of a perturbed problem is convex in the perturbed variables \cite{boyd2004convex}, we can conclude that $V_T(\mathbf{B}(T), \mathbf{r}(T),\mathbf{h}(T), \mathbf{w}(T))$ is convex in $(\mathbf{B}(T), \mathbf{r}(T))$. Moreover, $V_T(\cdot)$ is non-increasing in $\mathbf{B}(T)$ because more energy can only enable transmission of more bits, and it is non-decreasing in $\mathbf{r}(T)$ because  the larger the number of bits to be transmitted, the higher the objective value. 
	
	Consider the induction hypothesis that  $V_t\left(\mathbf{B}(t),\mathbf{r}(t), \mathbf{h}(t), \mathbf{r}(t)\right)$ is convex    non-decreasing in $\mathbf{r}(t)$ and convex    non-increasing in  $\mathbf{B}(t)$ for any $t\in \{T,T-1, \ldots,1\}$. Then, at time step $t$,  we solve \eqref{eq:at_t}. Due to linearity of expectation, we have, $\mathbb{E}[V_t\left(\mathbf{B}(t),\mathbf{r}(t), \mathbf{h}(t), \mathbf{r}(t)\right)]$ is convex, where $\mathbf{r}(t) =  \mathbf{r}(t-1)-\boldsymbol{\rho}(t-1)$ with probability  $\mathbb{P}(A_i(t)=0)$ and $\mathbf{r}(t) =  r_{\rm max}$ with probability $\mathbb{P}(A_i(t)=1)$. 
	From Section 3.2.4 in \cite{boyd2004convex}, we note that composition of a convex non-decreasing (respectively, non-increasing) function with a convex (respectively, concave) function results in a convex function. 	
	Since $\mathbb{E}[V_t\left(\mathbf{B}(t),\mathbf{r}(t), \mathbf{h}(t), \mathbf{r}(t)\right)]$ is convex non-decreasing in $\mathbf{r}(t)$ and because 
	$\mathbf{r}(t)$ is affine (hence, convex) in $\boldsymbol{\rho}(t-1)$, $\mathbb{E}[V_t\left(\mathbf{B}(t),\mathbf{r}(t), \mathbf{h}(t), \mathbf{r}(t)\right)]$ is convex in $\boldsymbol{\rho}(t-1)$. 
	Similarly,  $\mathbb{E}[V_t\left(\mathbf{B}(t),\mathbf{r}(t), \mathbf{h}(t), \mathbf{r}(t)\right)]$ is convex non-increasing in $\mathbf{B}(t)$ and  as $B_i(t)$ is a linear (hence concave) function in $P_i(t-1)$, we have that $\mathbb{E}[V_t\left(\mathbf{B}(t),\mathbf{r}(t), \mathbf{h}(t), \mathbf{r}(t)\right)]$ is convex in $\mathbf{P}(t-1)$.  Moreover, the first term, $f(r_i(t)-\rho_i(t))$ in \eqref{eq:Qfunction} is convex in $\rho_i(t)$. Hence, \eqref{eq:at_t} is a convex optimization problem.  Since the optimal objective value of a perturbed problem is convex in the perturbed variables \cite{boyd2004convex}, we conclude that  $V_{t-1}(\mathbf{B}(t-1), \mathbf{r}(t-1),\mathbf{h}(t-1), \mathbf{w}(t-1))$ is convex in $(\mathbf{B}(t-1),\mathbf{r}(t-1))$. The non-increasing and non-decreasing behavior of $V_{t-1}(\cdot)$ with respect to $\mathbf{B}(t-1)$, and $\mathbf{r}(t-1)$, respectively, follow due to the arguments presented for the $t=T$ case, above. 
Hence, the result follows via induction for $t=T-1,\ldots, 1$. 
\end{proof}

Based on the above theorem, we obtain the below result.  We define $\mathbf{x}_{-i}\triangleq (x_1,\ldots,x_{i-1}, x_{i+1},\ldots, x_M)$. 

\begin{theorem}\label{thm:monotonic}
	For given $\mathbf{h}(t)$, $\mathbf{w}(t)$, $\mathbf{B}_{-i}(t)$, $\mathbf{r}_{-i}(t)$, $\mathbf{P}_{-i}(t)$ and $\boldsymbol{\rho}_{-i}(t)$, the optimal solution to \eqref{eq:at_t}, $P_i(t)$ and $\rho_i(t)$ are non-decreasing in $B_i(t)$ for  fixed $r_i(t)$ and non-decreasing in $r_i(t)$ for  fixed $B_i(t)$. 
\end{theorem}
\begin{proof}
	We first prove the result for $t=T$. Consider $\mathbf{h}(T)$, $\mathbf{w}(T)$, $\mathbf{B}_{-i}(T)$, $\mathbf{r}(T)$, $\mathbf{P}_{-i}(T)$ and $\boldsymbol{\rho}_{-i}(T)$ to be fixed. From \eqref{eq:atT}, as $B_i(T)$ increases, the constraint in \eqref{eq:batT} becomes relaxed and hence the optimal objective does not decrease with increasing  $B_i(T)$,   which is possible only if $\rho_i(T)$ and  $P_i(T)$ do not decrease.  	 This  implies that $\rho_i(T)$ and $P_i(T)$ are non-decreasing in  $B_i(T)$. We now prove the same for $t<T$ in the below. The time indices of $B_i$, $B_i'$, $r_i$ and $r_i'$ are $t-1$ and that of $P_i$, $P_i'$, $\rho_i$ and $\rho_i'$ are $t-1$ and that of $V$ is $t$. For brevity, we drop the time indices and replace all the fixed quantities by $(\cdot)$.  Suppose that $B'_i\geq B_i$ and $P'_i\geq P_i$ and $\rho'_i\geq \rho_i$ for all $i\in \mathcal{M}$. We have 
	\begin{align}
	&Q(B_i', P_i', \rho_i',\cdot)-	Q(B_i', P_i, \rho_i,\cdot )\nonumber\\
	&=\sum_{i=1}^{M}\mathbb{E}[w_i f(r_i-\rho_i' )]+ \mathbb{E}[V\left(r_i-\rho_i' ,B_i'-P_i',\cdot \right)]\nonumber\\
	&\;\;-\sum_{i=1}^{M}\mathbb{E}[w_i f(r_i-\rho_i )]+ \mathbb{E}[V\left(r_i-\rho_i ,B_i'-P_i,\cdot \right)]\nonumber\\
		&\leq \sum_{i=1}^{M}\mathbb{E}[w_i f(r_i-\rho_i' )]+ \mathbb{E}[V\left(r_i-\rho_i' ,B_i-P_i',\cdot \right)]\nonumber\\
	&\;\;-\sum_{i=1}^{M}\mathbb{E}[w_i f(r_i-\rho_i )]+ \mathbb{E}[V\left(r_i-\rho_i ,B_i-P_i,\cdot \right)],\label{eq:ineq}
	\end{align}
	where the inequality in \eqref{eq:ineq} follows because 
	\begin{flalign}\label{eq:Vin}
	&\mathbb{E}[V\left(r_i-\rho_i' ,B_i'-P_i',\cdot \right)]-\mathbb{E}[V\left(r_i-\rho_i ',B_i-P_i',\cdot \right)]&&\nonumber\\
	&\leq \mathbb{E}[V\left(r_i-\rho_i ,B_i'-P_i,\cdot \right)] - \mathbb{E}[V\left(r_i-\rho_i ,B_i-P_i,\cdot \right)], &&
	\end{flalign}
	which is due to the following.  Recall that $V(\cdot)$ is convex in $B-P$. Moreover, $\rho'_i$ is a concave increasing function of $P'_i$, because $\mathbf{P}_{-i}$ and  $\boldsymbol{\rho}_{-i}$ are fixed and $g(\cdot)$ is concave increasing function of its argument.  
	Hence, $V(r_i-\rho_i',\cdot )$ is a composition of a convex nondecreasing function and a convex function in $P_i'$, and therefore, it is convex in $P_i'$ \cite{boyd2004convex}.  Noting that a convex function satisfies, $c(y+\delta)-c(y)\geq  c(x+\delta)-c(x)$ for any $y\geq x$ and $\delta\geq 0$ and letting, $y = B-P$, $x = B - P'$ and $\delta = B'-B$, as in the proof of Theorem 3 in \cite{Rui-MDP}, we get   \eqref{eq:Vin}. 
From \eqref{eq:ineq} and using the fact that the optimal $P_i(t-1)$ and $\rho_i(t-1)$ are obtained by minimizing $Q(\cdot)$ in \eqref{eq:at_t}, it can be seen that if it is optimal to transmit at power $P_i(t-1)$ when battery is $B_i(t-1)$, the transmit power $P_i'(t-1) \geq P_i(t-1)$, when battery increases to $B'_i(t-1)$. 
Using the similar arguments,   we can prove that  $P_i(t-1)$ and $\rho_i(t-1)$ are non-decreasing in $r_i(t-1)$ for  fixed $B_i(t-1)$. 
\end{proof}

Based on Theorem~\ref{thm:convex} and Theorem~\ref{thm:monotonic}, the numerical evaluation of the Bellman equations can be simplified. This is because,  for solving \eqref{eq:at_t} for any $B_i'(t)> B_i(t)$, the search space of $P_i'(t)$ and $\rho_i'(t)$ can be constrained to be greater than or equal to $P_i(t)$ and $\rho_i(t)$, respectively, which are the optimal solutions when the state of the battery is $B_i(t)$, when all other quantities in the  state and action are  fixed.  The results in Theorem~\ref{thm:convex} and Theorem~\ref{thm:monotonic} can also be used for efficiently solving the problem via approximate dynamic programming. 

\subsection{NN-Based Online Policy Derived from  the Offline Policy}
We now propose a heuristic policy by training a neural network (NN) to output action values, $(\mathbf{P}, \boldsymbol{\rho})$ by taking the states $(\mathbf{B}, \mathbf{r},\mathbf{h}, \mathbf{w})$ as inputs for each time instant. 
For this, we first solve \eqref{eq:main-opt-problem} in the offline case for each realization of the sample paths, assuming non-causal  knowledge of version and energy arrivals and channel power gain realizations for all $t\in \{1,\ldots,T\}$.   In the offline case, \eqref{eq:main-opt-problem} is a convex optimization problem
and it can be solved using standard numerical techniques. We obtain the optimal offline solution, $\mathbf{P}(t)$, $\boldsymbol{\rho}(t)$ for $\mathbf{B}(t), \mathbf{r}(t),\mathbf{h}(t), \mathbf{w}(t)$ for each $t\in \{1,2,\ldots,T\}$, which we consider as the training data. 
We run the offline policy for $N$ times to obtain $NT$ data samples. 
We consider the problem of learning actions for  each state as a regression problem, for which we adopt the mean squared error between the output of the neural network and the target values as the loss function.   
When an NN is trained, during the inference phase,  $\mathbf{\tilde{P}}$ and $\boldsymbol{\tilde{\rho}}$ values output by the NN may not satisfy the constraints in \eqref{eq:capacity}, \eqref{eq:power} and \eqref{eq:bits}. From them, we derive $\mathbf{P}$ and $\boldsymbol{\rho}$ values that satisfy the constraints as follows: We take  $P_i =  \min(\max({\tilde{P}_i},0),{B}_i)$ and ${\rho_i} =  \min(\max({\tilde{\rho}_i},0),\log(1+{h}_i{P}_i))$ for all $i\in \mathcal{M}$. We then verify if the obtained $\mathbf{{P}}$ and $\boldsymbol{{\rho}}$ values satisfy achievable rate region constraints, failing to which, we scale the $\boldsymbol{\rho}$ values appropriately to ensure that the constraints are satisfied. 

\subsection{Greedy Policy}
In this policy, we solve   problem  \eqref{eq:main-opt-problem} at each time slot  optimally. Across slots,  the number of remaining bits to be transmitted and the amount of energy in the battery evolve according to \eqref{eq:bit-evolution} and  \eqref{eq:battery-evolution}, respectively, as in the earlier policies. 




\begin{figure}[t]
%
%
\begin{tikzpicture}[scale=0.65]

\begin{axis}[%
width=4.521in,
height=3.566in,
at={(0in,0in)},
scale only axis,
xmin=0,
xmax=1,
xtick={0,0.1,	0.2,	0.3,	0.4,	0.5,	0.6,	0.7,	0.8,	0.9,	1},
xticklabels={0,0.1,	0.2,	0.3,	0.4,	0.5,	0.6,	0.7,	0.8,	0.9,	1},
xlabel style={font=\color{white!15!black}},
xlabel={Probability of Energy Arrival, $e_{\rm prob}$},
ymin=0.4,
ymax=1.65,
ylabel style={font=\color{white!15!black}},
ylabel={Cost, $\frac{1}{TM}\sum_{t=1}^{T}\sum_{i=1}^{M}\mathbb{E}[w_i(t)f(r_i(t)-\rho_i(t))]$},
axis background/.style={fill=white},
title style={font=\bfseries},
title={},
legend style={legend cell align=left, align=left, draw=white!15!black}
]

\addplot[color=red, thick, mark=square, mark options={solid, red}, mark size = 3pt] table [x=eprob, y=offline, col sep=comma] {eprob.csv};
\addlegendentry{Optimal Offline Policy};

\addplot[color=blue, thick,mark=o, mark options={solid, blue}, mark size = 3pt] table [x=eprob, y=MDP, col sep=comma] {eprob.csv};
\addlegendentry{MDP-Based Online Policy (Discrete State \& Action Spaces)};
 
\addplot[color=magenta, thick,mark=+, mark options={solid, magenta}, mark size = 3pt] table [x=eprob, y=NN, col sep=comma] {eprob.csv};
\addlegendentry{NN-Based Online Policy};
 
\addplot[color=green, thick,mark=diamond, mark options={solid, green}, mark size = 3pt] table [x=eprob, y=Greedy, col sep=comma] {eprob.csv};
\addlegendentry{Greedy Policy};

\end{axis}
\end{tikzpicture}%
	\caption{Variation of the objective value of \eqref{eq:main-opt-problem} with the probability of energy arrival, $e_{\rm prob}$ for $p_{\rm prob}=0.4$ and $i_{\rm prob} = 0.4$. Due to discretization of the state and action spaces, the MDP-based online policy performs worse than the NN-based online policy, which is trained to provide outputs over a continuous action space.}
	\label{fig:eprob}
\end{figure}
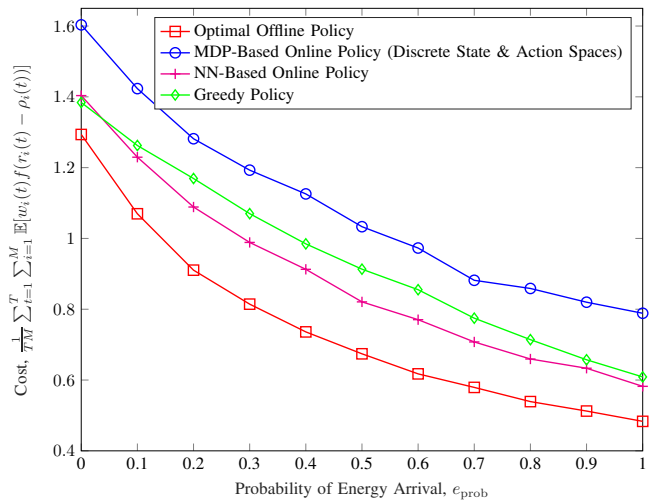 

\section{Numerical Results}\label{sec:num}
In this section, we present   numerical results for which we consider a  two-user case, with $M=2$. The channel power gains of both the users are assumed to be identically distributed over $\{0.1, 1\}$, where the probability of it being $0.1$ is equal to $0.4$.  In  a user, in a slot, a unit of energy arrives with probability $e_{\rm prob}$ or no energy arrives, and  a version carrying $4$ bits arrives with probability, $p_{\rm prob}$ or no version arrives.  When a version arrives, it has importance weight of $2$ and $1$  with probability, $i_{\rm prob}$ and $1-i_{\rm prob}$, respectively. 
For concreteness, we  consider $g(x)= \log(1+x)$ and $f(x) = \exp(x-4)$, motivated by the expressions for maximum achievable rate in an additive white Gaussian noise channel and the distortion incurred for a Gaussian source, respectively, in \cite{Cover}. Furthermore, we consider $B_{\rm max}=4$. For solving the MDP, we discretize the state of the battery, the transmit power and the number of transmitted bits  to take values in $\{0,1,2,3,4\}$ for both users.  

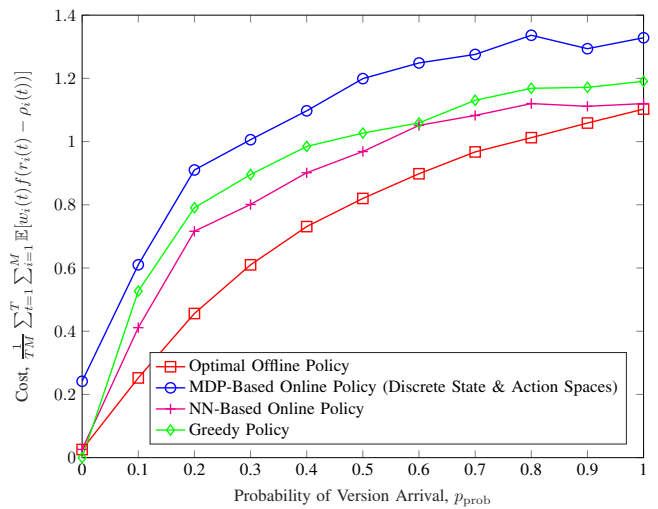
\begin{figure}[t]
%
%
\begin{tikzpicture}[scale=0.65]

\begin{axis}[%
width=4.521in,
height=3.566in,
at={(0in,0in)},
scale only axis,
xmin=0,
xmax=1,
xtick={0,0.1,	0.2,	0.3,	0.4,	0.5,	0.6,	0.7,	0.8,	0.9,	1},
xticklabels={0,0.1,	0.2,	0.3,	0.4,	0.5,	0.6,	0.7,	0.8,	0.9,	1},
xlabel style={font=\color{white!15!black}},
xlabel={Probability of Version  Arrival, $p_{\rm prob}$},
ymin=0,
ymax=1.4,
ylabel style={font=\color{white!15!black}},
ylabel={Cost, $\frac{1}{TM}\sum_{t=1}^{T}\sum_{i=1}^{M}\mathbb{E}[w_i(t)f(r_i(t) -\rho_i(t))]$},
axis background/.style={fill=white},
title style={font=\bfseries},
title={},
legend style={legend cell align=left, align=left, draw=white!15!black},
legend pos={south east}
]

\addplot[color=red, thick, mark=square, mark options={solid, red}, mark size = 3pt] table [x=pprob, y=offline, col sep=comma] {pprob.csv};
\addlegendentry{Optimal Offline Policy};

\addplot[color=blue, thick,mark=o, mark options={solid, blue}, mark size = 3pt] table [x=pprob, y=MDP, col sep=comma] {pprob.csv};
\addlegendentry{MDP-Based Online Policy (Discrete State \& Action Spaces)};

\addplot[color=magenta, thick,mark=+, mark options={solid, magenta}, mark size = 3pt] table [x=pprob, y=NN, col sep=comma] {pprob.csv};
\addlegendentry{NN-Based Online Policy};

\addplot[color=green, thick,mark=diamond, mark options={solid, green}, mark size = 3pt] table [x=pprob, y=Greedy, col sep=comma] {pprob.csv};
\addlegendentry{Greedy Policy};

\end{axis}
\end{tikzpicture}%
	\caption{Variation of the objective value of \eqref{eq:main-opt-problem} with the probability of version arrival, $p_{\rm prob}$   for $e_{\rm prob}=0.4$ and $i_{\rm prob} = 0.4$. }
	\label{fig:pprob}
\end{figure} 

\begin{table}[t]										
	\centering									
	\caption{The  objective value of \eqref{eq:main-opt-problem} for different probability of importance weight being equal to $2$, $i_{\rm prob}$, under different policies. }									
	\begin{tabular}{|c|p{1cm}|p{1cm}|p{1cm}|p{1cm}|}									
		\hline								
		$i_{\rm prob}$&  Optimal Offline  & NN-Based Online & MDP-Based Online  & Greedy       \\								
		\hline								
		$0$        &$0.52$ & $0.64$  &  $0.82$ &  $0.60$ \\								
		\hline								
		$0.2$      & $0.62$ & $0.76$  &  $0.96$ & $0.71$ \\								
		\hline								
		$0.4$      &  $0.73$  & $0.92$ & $1.15$ & $0.85$  \\								
		\hline								
		$0.6$      &  $0.84$ & $1.03$  &  $1.24$ & $0.98$ \\								
		\hline								
		$0.8$      &  $0.94$  & $1.16$  &  $1.39$ & $1.16$ \\								
		\hline								
		$1$        &  $1.06$ & $1.25$ & $1.44$ & $1.38$ \\								
		\hline								
	\end{tabular}									
	\label{tab:iprob}									
\end{table}

Using the above setup, in Fig.~\ref{fig:eprob} and Fig.~\ref{fig:pprob},  we present variation of the objective value in \eqref{eq:main-opt-problem} with the probability of energy and version arrivals, respectively. We also show the objective value of \eqref{eq:main-opt-problem} for different probability of importance weight being equal to $2$, $i_{\rm prob}$, under different policies in Table~\ref{tab:iprob}. From the table, as expected, the objective value of all the policies increases with  $i_{\rm prob}$. Similarly, from the figures, we note that as the probability of energy (version) arrival increases, the objective value decreases (increases) as expected. Moreover, the optimal offline policy performs better than the NN-based and MDP-based online policies and the greedy policy, as expected, due to the non-causal information about the realizations of random processes available in the offline policy. \emph{It may be surprising to observe that the performance of the MDP-based policy is worse than that of the greedy and  NN-based online policy that is trained to output actions for a given state. This is because, the MDP-based policy is obtained by discretizing the   state of the battery, the transmit power and the number of transmitted bits  to take values in $\{0,1,2,3,4\}$, for which it is optimal. However, the greedy and NN-based online policies take input states and output actions over continuous spaces, which is larger  than the set $\{0,1,2,3,4\}$. Such a result has also been observed in \cite{ICASSP_Mohit}.} This performance degradation due to discretization can be overcome by approximating the value functions, $V$, in \eqref{eq:at_t} by continuous functions and solving the problem via approximate dynamic programming. 

\section{Conclusions}\label{sec:conclusions}
\balance 
In this work, we considered energy harvesting users communicating version updates to a common access point over a fading multiple access channel.   
The version updates having random importance weights arrive at users and  when a new version arrives all the earlier versions become obsolete.
In this setting, we considered a minimization of a finite horizon expected average of the product of importance weight  and a convex increasing function of the number of un-transmitted bits, subject to an achievable rate-region constraint of the MAC and energy constraints at each user, so as to enable \emph{importance-aware} delivery of  as many bits, as soon as possible. 
We cast the problem in the MDP framework to obtain optimal online policy.  We showed that the optimal transmit power and the transmitted  number of bits are non-decreasing functions of the amount of energy stored in the battery of a user when all other state and action variables are fixed, and are non-decreasing functions of the number of remaining bits to be transmitted of a user, when all other state and action variables are fixed.  
We also obtained an NN-based heuristic online policy, that is trained using the optimal solution in the offline case when all the future realizations of the random processes are known non-causally. Via numerical simulations, we showed that the NN-based policy performs competitively with respect to the MDP-based online policy.

\bibliographystyle{ieeetr}
\bibliography{references}
	
\end{document}